\documentclass[letterpaper, 10 pt, conference]{ieeeconf}%
\usepackage{url}
\usepackage{amssymb}
\usepackage{amsmath}
\usepackage{amsfonts}
\usepackage[noadjust]{cite}
\usepackage{cite}
\usepackage{graphicx}
\usepackage{float}
\usepackage{amsmath}
\usepackage{amssymb}
\usepackage{indentfirst}
\usepackage{graphicx}
\usepackage{afterpage}
\usepackage{subfigure}
\usepackage{cite}
\usepackage{amsfonts}%
\providecommand{\U}[1]{\protect\rule{.1in}{.1in}}
\IEEEoverridecommandlockouts
\newtheorem{theorem}{Theorem}

\newtheorem{remark}{Remark}

\begin{document}

\title{{\LARGE \textbf{Decentralized 2-D Control of Vehicular Platoons\\under Limited
Visual Feedback}}}
\author{Christos K. Verginis, Charalampos P. Bechlioulis, Dimos V. Dimarogonas and Kostas J.
Kyriakopoulos \thanks{C. K. Verginis, C. P. Bechlioulis and K. J. Kyriakopoulos
are with the Control Systems Laboratory, School of Mechanical Engineering,
National Technical University of Athens, Athens 15780, Greece. D. V.
Dimarogonas is with the Centre for Autonomous Systems at Kungliga Tekniska
Hogskolan, Stockholm 10044, Sweden. Emails:
{{ chrisverginis@gmail.com, chmpechl@mail.ntua.gr, dimos@kth.se,
kkyria@mail.ntua.gr}}.} \thanks{This work was supported by the EU funded
project RECONFIG: Cognitive, Decentralized Coordination of Heterogeneous
Multi-Robot Systems via Reconfigurable Task Planning, FP7-ICT-600825,
2013-2016 and the Swedish Research Council (VR).}}%

%TCIMACRO{\TeXButton{Make Title}{\maketitle}}%
%BeginExpansion
\maketitle
%EndExpansion
%

%TCIMACRO{\TeXButton{Begin abstract}{\begin{abstract}}}%
%BeginExpansion
\begin{abstract}%
%EndExpansion

In this paper, we consider the two dimensional (2-D) predecessor-following
control problem for a platoon of unicycle vehicles moving on a planar surface. More
specifically, we design a decentralized kinematic control protocol, in the
sense that each vehicle calculates its own control signal based solely on local
information regarding its preceding vehicle, by its on-board camera, without
incorporating any velocity measurements. Additionally, the transient and
steady state response is a priori determined by certain designer-specified
performance functions and is fully decoupled by the number of vehicles
composing the platoon and the control gains selection. Moreover, collisions
between successive vehicles as well as connectivity breaks, owing to the
limited field of view of cameras, are provably avoided. Finally, an extensive
simulation study is carried out in the WEBOTS$^{\text{TM}}$ realistic simulator, clarifying the proposed control scheme and verifying its effectiveness.%

%TCIMACRO{\TeXButton{End abstract}{\end{abstract}}}%
%BeginExpansion
\end{abstract}%
%EndExpansion

\section{Introduction\label{sec:Introduction}}

During the last few decades, the 1-D longitudinal control problem of Automated
Highway Systems (AHS) has become an active research area in automatic control
(see \cite{1996Swaroop_Hedrick,2001Soo_Chong_Roh,2005Jovanovic_Bamieh,2009Barooah_Mehta_Hespanha,2012Lin_Fardad_Jovanovic}
and the references therein). Unlike human drivers that are not able to react quickly and accurately enough to follow each other in close proximity at high speeds, the safety and capacity of highways (measured in vehicles/lanes/time) is significantly increased when vehicles operate autonomously, forming large platoons at close spacing. However, realistic situations necessitate for 2-D motion on planar surfaces (see Fig. \ref{Fig:platoons}).

Early works in \cite{1994Hedrick_Tomizuka_Varaiya,1997Li_Roberto_et_al,2000_Rajamani,1998Tan_Rajesh_Zhang}
consider the lane-keeping and lane-changing control for platoons in AHS,
adopting however a centralized network, where all vehicles exchange
information with a central computer that determines the control protocol,
making thus the overall system sensitive to delays, especially when a large
number of vehicles is involved. Alternatively, rigid multi-agent formations
are employed in decentralized control schemes, where each vehicle utilizes
relative information from its neighbors. The majority of these works consider
unicycle
\cite{2004_Mazo_Speranazon,2005_Mariottini_Pappas,2005_Gustavi,2002_Das_Fierro,2008_Gustavi_Hu}
and bicycle kinematic models
\cite{2006_Khatir_Davison,2004_Pham_Wang,2013Ali_Garcia_Martinet}. However,
many of them adopt linearization techniques
\cite{2005_Mariottini_Pappas,2004_Tanner_Pappas,2003_Lawton_Beard,2002_Das_Fierro,2014_Maithripala_Berg,2006_Khatir_Davison,2013Ali_Garcia_Martinet,1994Godbole_Lygeros}
that may lead to unstable inner dynamics or degenerate configurations owing to
the non-holonomic constraints of the vehicles, as shown in \cite{2003_Vidal}.
Additionally, each vehicle is assumed to have access to the neighboring
vehicles' velocity, either explicitly, hence degenerating the decentralized
form of the system and imposing inherent communication delays, or by
employing observers \cite{2008_Gustavi_Hu} that increase the overall design
complexity. Furthermore, the transient and steady state response of the closed
loop is affected severely by the control gains selection
\cite{2009_Hao_Barooah}, thus limiting the controller's robustness and
complicating the design procedure.

\begin{figure}[t]
\begin{center}
\includegraphics[width=3.35in]{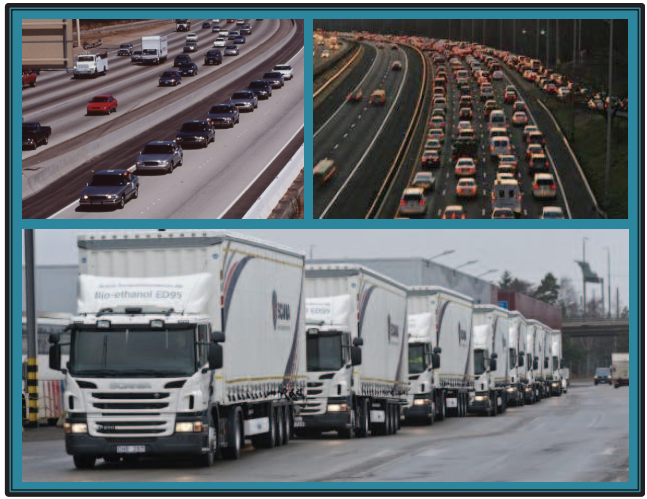}
\end{center}
\caption{Vehicular platoons in 2-D motion on planar roads.}%
\label{Fig:platoons}%
\end{figure}

Another significant issue affecting the $2$-D control of vehicular platoons
concerns the sensing capabilities when visual feedback from cameras is employed. A vast number of the related works neglects
the sensory limitations, which however are crucial in real-time scenarios. In
\cite{2002_Das_Fierro,2003_Vidal} visual feedback from omnidirectional cameras
is adopted, not accounting thus for sensor limitations, which however are
examined in \cite{2004_Mazo_Speranazon} considering directional sensors for
the tracking problem of a moving object by a group of robots. Although cameras
are directional sensors, they inherently have a limited range and a limited
angle of view as well. Hence, in such cases each agent should keep a certain
close distance and heading angle from its neighbors, in order to avoid
connectivity breaks. Thus, it is clear that limited sensory capabilities lead
to additional constraints on the behavior of the system, that should therefore
be taken into account exclusively when designing the control protocols. The aforementioned specifications were considered in \cite{2014_Panagou_Kumar}, where a solution based on set-theory and dipolar vector fields was introduced. Alternatively, a visual-servoing scheme for leader-follower formation was presented in \cite{2003_Cowan_et_al}. Finally, a centralized control protocol under vision-based localization for leader-follower formations was adopted in \cite{2007_Mariottini_et_al, 2009_Mariottini_et_al}.

In this paper, we extend our previous work on 1-D longitudinal control of
vehicular platoons \cite{2014Bechlioulis_Dimarogonas_Kyriakopoulos} to $2$-D motion on planar surfaces, under the predecessor-following
architecture. We design a fully decentralized kinematic control protocol, in
the sense that each vehicle has access only to the relative distance and
heading error with respect to its preceding vehicle. Such information is
obtained by an onboard camera with limited field of view \cite{2005_Gustavi}, that imposes inevitably certain constraints on the
configuration of the platoon. More specifically, each vehicle aims at
achieving a desired distance from its predecessor, while keeping it within the
field of view of its onboard camera in order to maintain visual connectivity
and avoid collisions. Moreover, the transient and steady state response is
fully decoupled by the number of vehicles and the control gains selection.
Finally, the explicit collision avoidance and connectivity maintenance
properties are imposed by certain designer-specified performance functions,
that incorporate the aforementioned visual constraints. In summary, the main
contributions of this work are given as follows:

\begin{itemize}
\item We propose a novel solution to the $2$-D formation control problem of vehicular
platoons, avoiding collisions and connectivity breaks owing to visual feedback constraints.

\item We develop a fully decentralized kinematic control protocol, in the
sense that the feedback of each vehicle is based exclusively on its own
camera, without incorporating any measurement of the velocity of the preceding vehicle.

\item The transient and steady state response of the closed loop system is
explicitly determined by certain designer-specified performance functions,
simplifying thus the control gains selection.
\end{itemize}

The manuscript is organized as follows. The problem statement is given in
Section \ref{PROBLEM STATEMENT}. The decentralized control protocol is
provided in Section \ref{Control Design}. In Section
\ref{sec:simulation_results}, an extensive simulation study is presented,
clarifying and verifying the theoretical findings. Finally, we conclude in
Section \ref{sec:conclusions}.

\section{{Problem Statement\label{PROBLEM STATEMENT}}}

Consider a platoon of $N$ vehicles moving on a planar surface under unicycle
kinematics:%
\begin{equation}
\left.
\begin{array}
[c]{l}%
\dot{x}_{i}=v_{i}\cos\varphi_{i}\\
\dot{y}_{i}=v_{i}\sin\varphi_{i}\\
\dot{\varphi}_{i}=\omega_{i}%
\end{array}
\right\}  \text{, }i=1,\dots,N\label{eqn:Dynamics}%
\end{equation}
where $x_{i}$, $y_{i}$, $\varphi_{i}$ denote the position and orientation of
each vehicle on the plane and $v_{i},\omega_{i}$ are the linear and angular
velocities respectively. Let us also denote by $d_{i}(t)$ and $\beta_{i}(t)$
the distance and the bearing angle between successive vehicles $i$ and $i-1$
(see Fig. \ref{Fig:Formation}). Furthermore, we assume that the only available
feedback concerns the distance $d_{i}(t)$ and the bearing angle $\beta_{i}(t)$,
which both emanate from an onboard camera that detects a specific marker on
the preceding vehicle (e.g., the number plate). The control objective is to
design a distributed control protocol based exclusively on visual feedback
such that $d_{i}(t)\rightarrow d_{i,des}$ and $\beta_{i}(t)\rightarrow0$,
i.e., each vehicle tracks its predecessor and maintains a prespecified desired
distance $d_{i,des}$. Additionally, $d_{i}(t)$ should be kept greater than
$d_{\mathsf{col}}$ to avoid collisions between successive vehicles. In the
same vein, the inter-vehicular distance $d_{i}(t)$\ and the bearing angle
$\beta_{i}(t)$ should be kept less than $d_{\mathsf{con}}>d_{\mathsf{col}}$ and $\beta_{\mathsf{con}}$ respectively, in
order to maintain the connectivity owing to the camera's limited field
of view (see Fig. \ref{Fig:Formation}). Moreover, the desired trajectory of
the formation is generated by a reference/leading unicycle vehicle:%
\[
\begin{array}
[c]{l}%
\dot{x}_{0}=v_{0}\cos\varphi_{0}\\
\dot{y}_{0}=v_{0}\sin\varphi_{0}\\
\dot{\varphi}_{0}=\omega_{0}%
\end{array}
\]
with bounded velocities $v_{0}(t)$, $\omega_{0}(t)$ and is only provided to
the first vehicle. Finally, to solve the aforementioned control problem, we
assume that initially each vehicle lies within the field of view of its
follower's camera and no collision occurs, which are formulated as follows.

\begin{figure}[ptb]
\begin{center}
\includegraphics[width=3.35in]{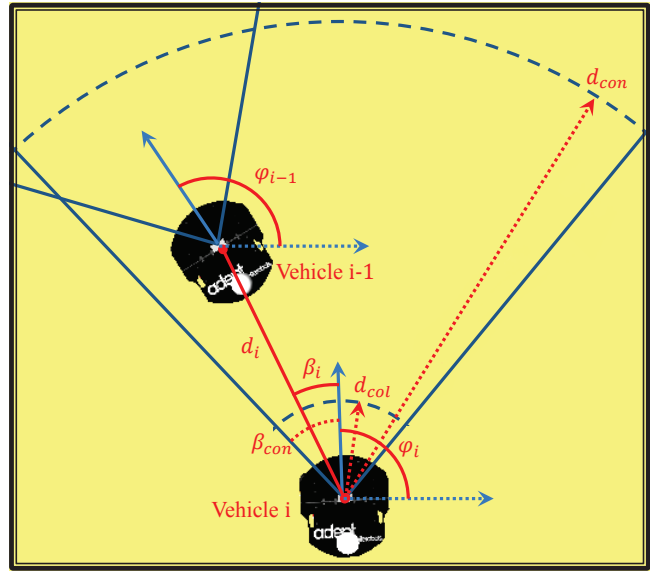}
\end{center}
\caption{Graphical illustration of two successive vehicles in the platoon.
Each vehicle should keep its distance $d_{i}\left(  t\right)$ and bearing angle $\beta_{i}\left(  t\right)$ to its predecessor within the feasible area
$d_{\mathsf{col}}<d_{i}\left(  t\right)<d_{\mathsf{con}}$ and $\left|\beta_{i}\left(t\right)\right|<\beta_{\mathsf{con}}$, thus avoiding collisions and connectivity breaks.}%
\label{Fig:Formation}%
\end{figure}

\textbf{Assumption A1.} The initial state of the platoon does not violate the
collision and connectivity constraints, i.e., $d_{\mathsf{col}}<d_{i}%
(0)<d_{\mathsf{con}\text{ }}$ and $\left\vert \beta_{i}(0)\right\vert
<\beta_{\mathsf{con}}$, $i=1,\dots,N$.

In the sequel, we define the distance and heading errors:%
\begin{equation}
\left.
\begin{array}
[c]{l}%
e_{d_{i}}(t)=d_{i}(t)-d_{i,des}\\
e_{\beta_{i}}(t)=\beta_{i}(t)
\end{array}
\right\}  \text{, }i=1,\dots,N \label{eqn:errors}%
\end{equation}
where $d_{i}\left(  t\right)  =\sqrt{\left(  x_{i}\left(  t\right)
-x_{i-1}\left(  t\right)  \right)  {}^{2}+\left(  y_{i}\left(  t\right)
-y_{i-1}\left(  t\right)  \right)  {}^{2}}$. Hence, differentiating
(\ref{eqn:errors}) with respect to time and substituting (\ref{eqn:Dynamics}),
we obtain:%
\begin{equation}
\left.
\begin{array}
[c]{l}%
\dot{e}_{d_{i}}=-v_{i}\cos\beta_{i}+v_{i-1}\cos(\gamma_{i}+\beta_{i})\\
\dot{e}_{\beta_{i}}=-\omega_{i}+\frac{v_{i}}{d_{i}}\sin\beta_{i}-\frac
{v_{i-1}}{d_{i}}\sin(\gamma_{i}+\beta_{i})
\end{array}
\right\}  \text{, }i=1,\dots,N \label{eqn:error_dynamics}%
\end{equation}
where $\gamma_{i}(t)=\varphi_{i}(t)-\varphi_{i-1}(t)$, which may be expressed
in vector form as follows:%
\begin{equation}%
\begin{array}
[c]{l}%
\dot{e}_{d}=-\tilde{C}v+c\\
\dot{e}_{\beta}=-\omega+D^{-1}(\tilde{S}v+s)
\end{array}
\label{eqn:platoon_error_dynamics}%
\end{equation}
where%
\begin{align*}
e_{d}  &  =\left[  e_{d_{1}},\dots,e_{d_{N}}\right]  ^{T}\text{, }e_{\beta
}=\left[  e_{\beta_{1}},\dots,e_{\beta_{N}}\right]  ^{T}\\
v  &  =[v_{1},\dots,v_{N}]^{T}\text{, }\omega=[\omega_{1},\dots,\omega
_{N}]^{T}\\
D  &  =\mathrm{diag}(d_{1},\dots,d_{N})^{T}\\
c  &  =[v_{0}\cos(\gamma_{1}+\beta_{1}),0,\dots,0]^{T}\\
s  &  =[v_{0}\sin(\gamma_{1}+\beta_{1}),0,\dots,0]^{T}%
\end{align*}
and $\tilde{C}$, $\tilde{S}$ are the lower bi-diagonal matrices:%
\[%
\begin{array}
[c]{l}%
{\tilde{C}}{=}\left[
\begin{array}
[c]{cccc}%
\cos\beta_{1} & 0 & \cdots & 0\\
-\cos(\beta_{2}+\gamma_{2}) & \cos\beta_{2} &  & \vdots\\
0 & \ddots & \ddots & \\
&  & \ddots & 0\\
0 & \cdots & -\cos(\beta_{N}+\gamma_{N}) & \cos\beta_{N}%
\end{array}
\right] \\
\\
{\tilde{S}}{=}\left[
\begin{array}
[c]{cccc}%
\sin\beta_{1} & 0 & \cdots & 0\\
-\sin(\beta_{2}+\gamma_{2}) & \sin\beta_{2} &  & \vdots\\
0 & \ddots & \ddots & \\
&  & \ddots & 0\\
0 & \cdots & -\sin(\beta_{N}+\gamma_{N}) & \sin\beta_{N}%
\end{array}
\right].
\end{array}
\]

\section{{Control Design\label{Control Design}}}

The concepts and techniques in the scope of prescribed performance
control, recently proposed in \cite{2014Bechlioulis_Rovithakis},
are adapted in this work in order to: i)
achieve predefined transient and steady state response for the distance and
heading errors $e_{d_{i}}(t)$, $e_{\beta_{i}}(t)$, $i=1,\ldots,N$ as well as
ii) avoid the violation of the collision and connectivity constraints
presented in Section \ref{PROBLEM STATEMENT}. As stated in \cite{2014Bechlioulis_Rovithakis}, prescribed performance characterizes the behavior where the
aforementioned errors evolve strictly within a predefined region that is
bounded by absolutely decaying functions of time, called performance
functions. The mathematical expressions of prescribed performance is given by
the following inequalities:%
\begin{equation}
\left.
\begin{array}
[c]{l}%
-\underline{M}_{d_{i}}\rho_{d_{i}}\left(  t\right)  <e_{d_{i}}\left(
t\right)  <\overline{M}_{d_{i}}\rho_{d_{i}}\left(  t\right) \\
-\underline{M}_{\beta_{i}}\rho_{\beta_{i}}\left(  t\right)  <e_{\beta_{i}%
}\left(  t\right)  <\overline{M}_{\beta_{i}}\rho_{\beta_{i}}\left(  t\right)
\end{array}
\right\}  \text{, }{i=1,\ldots,N} \label{eqn:PP_inequalities}%
\end{equation}
{for all }$t\geq0${, where}%
\begin{equation}%
\begin{array}
[c]{l}%
\rho_{d_{i}}(t)=(1-\tfrac{\rho_{d,\infty}}{\max\left\{  \underline{M}_{d_{i}%
},\overline{M}_{d_{i}}\right\}  })e^{-l_{d}t}+\tfrac{\rho_{d,\infty}}%
{\max\left\{  \underline{M}_{d_{i}},\overline{M}_{d_{i}}\right\}  }\\
\rho_{\beta_{i}}(t)=(1-\tfrac{\rho_{\beta,\infty}}{\max\left\{  \underline
{M}_{\beta_{i}},\overline{M}_{\beta_{i}}\right\}  })e^{-l_{\beta}t}%
+\tfrac{\rho_{\beta,\infty}}{\max\left\{  \underline{M}_{\beta_{i}}%
,\overline{M}_{\beta_{i}}\right\}  }%
\end{array}
\label{eqn:PP_rho}%
\end{equation}
{are designer-specified, smooth, bounded and decreasing positive functions of
time with positive parameters $l_{j}$, $\rho_{j,\infty}$, $j\in\left\{  d,\beta\right\}$ incorporating the desired transient and steady state
performance respectively, and $\underline{M}_{j_{i}}$, $\overline{M}_{j_{i}}$,
$j\in\left\{  d,\beta\right\}  $, $i=1,\dots,N$ are positive parameters
selected appropriately to satisfy the collision and connectivity constraints,
as presented in the sequel. In particular, the decreasing rate of $\rho
_{j_{i}}\left(  t\right)$, $j\in\left\{  d,\beta\right\}$, $i=1,\dots,N$, which is affected by the constant $l_{j}$, $j\in\left\{  d,\beta\right\}$
introduces a lower bound on the speed of convergence of $e_{j_{i}}\left(
t\right)  $, $j\in\left\{  d,\beta\right\}$, $i=1,\dots,N$. Furthermore, the
constants $\rho_{j,\infty}$, $j\in\left\{  d,\beta\right\}$ can be set arbitrarily small (i.e., $\rho
_{j,\infty}\ll\max\left\{  \underline{M}_{j_{i}},\overline{M}_{j_{i}}\right\}$, $j\in\left\{  d,\beta\right\}$, $i=1,\dots,N$), thus achieving practical convergence of the distance and heading errors to zero. Additionally,
we select: }

{
\begin{equation}
\left.
\begin{array}
[c]{l}%
\underline{M}_{d_{i}}=d_{i,des}-d_{\mathsf{col}}\\
\overline{M}_{d_{i}}=d_{\mathsf{con}}-d_{i,des}\\
\underline{M}_{\beta_{i}}=\overline{M}_{\beta_{i}}=\beta_{\mathsf{con}}%
\end{array}
\right\}  \text{, }i=1,\dots,N\text{.} \label{eqn:Md_Mbeta}%
\end{equation}
Notice that the parameters $d_{\mathsf{con}}$, $\beta
_{\mathsf{con}}$ are related to the constraints imposed by the camera's
limited field of view. More specifically, $d_{\mathsf{con}}$ should be
assigned a value less or equal to the distance from which the marker on the preceding vehicle may be
detected by the follower's camera, whereas $\beta_{\mathsf{con}}$ should be less or equal
to the half of the camera's angle of view, from which it follows that
$\beta_{\mathsf{con}}<\frac{\pi}{2}$ for common cameras. Apparently, since the
desired formation is compatible with the collision and connectivity
constraints (i.e., $d_{\mathsf{col}}<d_{i,des}<d_{\mathsf{con}}$,
$i=1,\dots,N$), the aforementioned selection ensures that $\underline
{M}_{j_{i}}$, $\overline{M}_{j_{i}}$ $>0$, $j\in\left\{  d,\beta\right\}  $,
$i=1,\dots,N$ and consequently under \textbf{Assumption A1} that:}%
\begin{equation}
\left.
\begin{array}
[c]{l}%
-\underline{M}_{d_{i}}\rho_{d_{i}}\left(  0\right)  <e_{d_{i}}\left(
0\right)  <\overline{M}_{d_{i}}\rho_{d_{i}}\left(  0\right) \\
-\underline{M}_{\beta_{i}}\rho_{\beta_{i}}\left(  0\right)  <e_{\beta_{i}%
}\left(  0\right)  <\overline{M}_{\beta_{i}}\rho_{\beta_{i}}\left(  0\right)
\end{array}
\right\}  {,~i=1,\ldots,N}\text{.} \label{eqn:Con. initial}%
\end{equation}
{Hence, guaranteeing prescribed performance via (\ref{eqn:PP_inequalities})
for all }$t>0${ and employing the decreasing property of $\rho_{j_{i}}(t)$,
$j\in\left\{  d,\beta\right\}  $,~$i=1,\ldots,N$, we conclude:}%
\[
\left.
\begin{array}
[c]{l}%
-\underline{M}_{d_{i}}<e_{d_{i}}\left(  t\right)  <\overline{M}_{d_{i}}\\
-\underline{M}_{\beta_{i}}<e_{\beta_{i}}\left(  t\right)  <\overline{M}%
_{\beta_{i}}%
\end{array}
\right\}  {,~i=1,\ldots,N}%
\]
{and consequently, owing to (\ref{eqn:Md_Mbeta}):}%
\[
\left.
\begin{array}
[c]{r}%
d_{\mathsf{col}}<d_{i}(t)<d_{\mathsf{con}}\\
-\beta_{\mathsf{con}}<\beta_{i}(t)<\beta_{\mathsf{con}}%
\end{array}
\right\}  {,~i=1,\ldots,N}%
\]
{for all }$t\geq0$, {which ensures the satisfaction of the collision and
connectivity constraints. }

\subsection{{Decentralized Control Protocol}}

{In the sequel, we propose a decentralized control protocol that guarantees
(\ref{eqn:PP_inequalities}) for all $t\geq0$, thus leading to the solution of
the 2-D formation control problem with prescribed performance under collision
and connectivity constraints for the considered platoon of vehicles. Hence,
given the distance and heading errors $e_{j_{i}}\left(  t\right)  $,
$j\in\left\{  d,\beta\right\}  $, $i=1,\dots,N$ defined in (\ref{eqn:errors}): }

{\textbf{Step I.} \textit{Select the corresponding performance functions }$\rho_{j_{i}%
}\left(  t\right)  $\textit{\ and positive parameters }$\underline{M}_{j_{i}%
},\overline{M}_{j_{i}}$\textit{, }$j\in\left\{  d,\beta\right\}  $%
,\textit{\ }$i=1,\dots,N$\textit{\ following (\ref{eqn:PP_rho}) and
(\ref{eqn:Md_Mbeta}) respectively, that incorporate the desired transient and
steady state performance specifications as well as the collision and
connectivity constraints.} }

{\textbf{Step II.} \textit{Define the normalized errors as:}}%
\begin{align}
\xi_{d}\left(  e_{d},t\right)   &  =\left[
\begin{array}
[c]{c}%
\xi_{d_{1}}\left(  e_{d_{1}},t\right) \\
\vdots\\
\xi_{d_{N}}\left(  e_{d_{N}},t\right)
\end{array}
\right]  :=\left[
\begin{array}
[c]{c}%
\tfrac{e_{d_{1}}}{\rho_{d_{1}}\left(  t\right)  }\\
\vdots\\
\tfrac{e_{d_{N}}}{\rho_{d_{N}}\left(  t\right)  }%
\end{array}
\right]  \triangleq\left(  \rho_{d}\left(  t\right)  \right)  ^{-1}%
e_{d}\label{eqn:ksid}\\
\xi_{\beta}\left(  e_{\beta},t\right)   &  =\left[
\begin{array}
[c]{c}%
\xi_{\beta_{1}}\left(  e_{\beta_{1}},t\right) \\
\vdots\\
\xi_{\beta_{N}}\left(  e_{\beta_{N}},t\right)
\end{array}
\right]  :=\left[
\begin{array}
[c]{c}%
\tfrac{e_{\beta_{1}}}{\rho_{\beta_{1}}\left(  t\right)  }\\
\vdots\\
\tfrac{e_{\beta_{N}}}{\rho_{\beta_{N}}\left(  t\right)  }%
\end{array}
\right]  \triangleq\left(  \rho_{\beta}\left(  t\right)  \right)
^{-1}e_{\beta} \label{eqn:ksib}%
\end{align}
{\textit{where }$\rho_{j}\left(  t\right)  =\mathrm{diag}\left(  \left[
\rho_{j_{i}}\left(  t\right)  \right]  _{i=1,\dots,N}\right)  $,
\textit{\ }$j\in\left\{  d,\beta\right\}  $, \textit{and design the
decentralized control protocol as:}%
\begin{equation}
v\left(  \xi_{d},t\right)  =\left[
\begin{array}
[c]{c}%
v_{1}\left(  \xi_{d_{1}},t\right) \\
\vdots\\
v_{_{N}}\left(  \xi_{d_{N}},t\right)
\end{array}
\right]  =K_{d}\varepsilon_{d}\left(  \xi_{d}\right)  \label{eqn:v_des}%
\end{equation}%
\begin{equation}
\omega\left(  \xi_{\beta},t\right)  =\left[
\begin{array}
[c]{c}%
\omega_{1}\left(  \xi_{\beta_{1}},t\right) \\
\vdots\\
\omega_{N}\left(  \xi_{\beta_{N}},t\right)
\end{array}
\right]  =K_{\beta}\left(  \rho_{\beta}\left(  t\right)  \right)
^{-1}r_{\beta}\left(  \xi_{\beta}\right)  \varepsilon_{\beta}\left(
\xi_{\beta}\right)  \label{eqn:wmega_des}%
\end{equation}
with $K_{j}=\mathrm{diag}(k_{j_{1},}\dots,k_{j_{N}})$, $k_{j_{i}}>0$,
$j\in\left\{  d,\beta\right\}  $, $i=1,\dots,N$, and }

{
\begin{align}
r_{\beta}\left(  \xi_{\beta}\right)   &  =\mathrm{diag}\left(  \left[
\tfrac{\frac{1}{\underline{M}_{\beta_{i}}}+\frac{1}{\overline{M}_{\beta_{i}}}%
}{\left(  1+\frac{\xi_{\beta_{i}}}{\underline{M}_{\beta_{i}}}\right)  \left(
1-\frac{\xi_{\beta_{i}}}{\overline{M}_{\beta_{i}}}\right)  }\right]
_{i=1,\dots,N}\right) \label{eqn: r_beta}\\
\varepsilon_{d}\left(  \xi_{d}\right)   &  =\left[  \ln\left(  \tfrac
{1+\frac{\xi_{d_{1}}}{\underline{M}_{d_{1}}}}{1-\frac{\xi_{d_{1}}}%
{\overline{M}_{d_{1}}}}\right)  ,\dots,\ln\left(  \tfrac{1+\frac{\xi_{d_{N}}%
}{\underline{M}_{d_{N}}}}{1-\frac{\xi_{d_{N}}}{\overline{M}_{d_{N}}}}\right)
\right]  ^{T}\label{eqn:epsilon_d}\\
\varepsilon_{\beta}\left(  \xi_{\beta}\right)   &  =\left[  \ln\left(
\tfrac{1+\frac{\xi_{\beta_{1}}}{\underline{M}_{\beta_{1}}}}{1-\frac{\xi
_{\beta_{1}}}{\overline{M}_{\beta_{1}}}}\right)  ,\dots,\ln\left(
\tfrac{1+\frac{\xi_{\beta_{N}}}{\underline{M}_{\beta_{N}}}}{1-\frac{\xi
_{\beta_{N}}}{\overline{M}_{\beta_{N}}}}\right)  \right]  ^{T}\text{.}
\label{eqn:epsilon_beta}%
\end{align}
}

\begin{remark}
{Notice from (\ref{eqn:v_des}) and (\ref{eqn:wmega_des}) that the proposed
control protocol is decentralized in the sense that each vehicle utilizes only
local relative to its preceding vehicle information, obtained by its on board
camera, to calculate its own control signal. Furthermore, the proposed
methodology results in a low complexity design. No hard
calculations (neither analytic nor numerical) are required to output the
proposed control signal, thus making its distributed implementation
straightforward. Additionally, we stress that the desired transient and steady
state performance specifications as well as the collision and connectivity
constraints are exclusively introduced via the appropriate selection of
$\rho_{j_{i}}\left(  t\right)  $ and $\underline{M}_{j_{i}},\overline
{M}_{j_{i}}$, $j\in\left\{  d,\beta\right\}  $, $i=1,\dots,N$. }
\end{remark}

\subsection{{Stability Analysis\label{Stability}}}

The main results of this work are summarized in the following theorem.

\begin{theorem}
Consider a platoon of $N$ unicycle vehicles aiming at establishing a formation
described by the desired inter-vehicular distances $d_{i,des}${, $i=1,\dots
,N$, while satisfying the collision and connectivity constraints represented
by }$d_{\mathsf{col}}$ and $d_{\mathsf{con}}$, $\beta_{\mathsf{con}}$
respectively, with { $d_{\mathsf{col}}<d_{i,des}<d_{\mathsf{con}}$,
$i=1,\dots,N$ and $\beta_{\mathsf{con}}<\frac{\pi}{2}$. Under }%
\textbf{Assumption A1}, the decentralized control protocol (\ref{eqn:ksid}%
)-(\ref{eqn:epsilon_beta}) guarantees:%
\[
\left.
\begin{array}
[c]{l}%
-\underline{M}_{d_{i}}\rho_{d_{i}}\left(  t\right)  <e_{d_{i}}\left(
t\right)  <\overline{M}_{d_{i}}\rho_{d_{i}}\left(  t\right)  \\
-\underline{M}_{\beta_{i}}\rho_{\beta_{i}}\left(  t\right)  <e_{\beta_{i}%
}\left(  t\right)  <\overline{M}_{\beta_{i}}\rho_{\beta_{i}}\left(  t\right)
\end{array}
\right\}  \text{, }{i=1,\ldots,N}%
\]
for all { $t\geq0$, as well as the boundedness of all closed loop signals.}
\end{theorem}

\begin{proof}
Differentiating (\ref{eqn:ksid}) and (\ref{eqn:ksib}) with respect to time, we
obtain:%
\begin{align}
\dot{\xi}_{d} &  =(\rho_{d}\left(  t\right)  )^{-1}(\dot{e}_{d}-\dot{\rho}%
_{d}\left(  t\right)  \xi_{d})\label{eqn:ksi_d_dot}\\
\dot{\xi}_{\beta} &  =(\rho_{\beta}\left(  t\right)  )^{-1}(\dot{e}_{\beta
}-\dot{\rho}_{\beta}\left(  t\right)  \xi_{\beta})\label{eqn:ksi_beta_dot}%
\end{align}
Employing (\ref{eqn:platoon_error_dynamics}), (\ref{eqn:v_des}) and
(\ref{eqn:wmega_des}), we arrive at:%
\begin{align}
\dot{\xi}_{d} &  =h_{d}(t,\xi_{d})\nonumber\\
&  =(\rho_{d}\left(  t\right)  )^{-1}(-\tilde{C}K_{d}\varepsilon_{d}\left(
\xi_{d}\right)  +c-\dot{\rho}_{d}\left(  t\right)  \xi_{d}%
)\label{eqn:ksi_d_dot_2}\\
\dot{\xi}_{\beta} &  =h_{\beta}(t,\xi_{d},\xi_{\beta})\nonumber\\
&  =(\rho_{\beta}\left(  t\right)  )^{-1}(-K_{\beta}(\rho_{\beta}\left(
t\right)  )^{-1}r_{\beta}\left(  \xi_{\beta}\right)  \varepsilon_{\beta
}\left(  \xi_{\beta}\right)\nonumber\\
&  \left.  \text{ \ \ \ }+D^{-1}\tilde{S}K_{d}\varepsilon_{d}\left(  \xi_{d}\right) %
  +D^{-1}s-\dot{\rho}_{\beta}\left(  t\right)  \xi_{\beta}\right)
.\label{eqn:ksi_beta_dot_2}%
\end{align}
Thus, the closed loop dynamical system of $\xi(t)=\left[  \xi_{d}^{T}%
(t),\xi_{\beta}^{T}(t)\right]  ^{T}$ may be written in compact form as:
\begin{equation}
\dot{\xi}=h(t,\xi)=\left[
\begin{array}
[c]{c}%
h_{d}(t,\xi_{d})\\
h_{\beta}(t,\xi_{d},\xi_{\beta})
\end{array}
\right]  .\label{eqn:ksi_dot}%
\end{equation}
Let us also define the open set $\Omega_{\xi}=\Omega_{\xi_{d}}\times
\Omega_{\xi_{\beta}}$, where:%
\begin{align}
\Omega_{\xi_{d}} &  =(-\underline{M}_{d_{1}},\overline{M}_{d_{1}})\times
\dots\times(-\underline{M}_{d_{N}},\overline{M}_{d_{N}})\nonumber\\
\Omega_{\xi_{\beta}} &  =(-\underline{M}_{\beta_{1}},\overline{M}_{\beta_{1}%
})\times\dots\times(-\underline{M}_{\beta_{N}},\overline{M}_{\beta_{N}%
})\text{.}\nonumber
\end{align}
In what follows, we proceed in two phases. First, the existence of a unique
solution $\xi(t)$ of (\ref{eqn:ksi_dot}) over the set $\Omega_{\xi}$ for a
time interval $\left[  0,\tau_{\max}\right)  $ is ensured (i.e., $\xi(t)$
$\in\Omega_{\xi},\forall t\in\left[  0,\tau_{\max}\right)  $). Then, we prove
that the proposed control protocol (\ref{eqn:ksid})-(\ref{eqn:epsilon_beta}) guarantees: a) the boundedness of all closed loop signals for all $t\in\left[
0,\tau_{\max}\right)  $ as well as that b) $\xi(t)$ remains strictly within a
compact subset of $\Omega_{\xi}$, which leads by contradiction to $\tau_{\max
}=\infty$ and consequently to the completion of the proof.

\textit{Phase A.} Selecting the parameters $\underline{M}_{j_{i}},\overline
{M}_{j_{i}},${ $j\in\left\{  d,\beta\right\}  $, $i=1,\dots,N$ according to
(\ref{eqn:Md_Mbeta}), we guarantee that the set }$\Omega_{\xi}$ is nonempty
and open. Moreover, as shown in (\ref{eqn:Con. initial}) from
\textbf{Assumption A1}, we conclude that\textbf{ }$\xi(0)\in\Omega_{\xi}$.
Additionally, notice that the function $h$ is continuous in $t$ and locally
Lipschitz in $\xi$ over the set $\Omega_{\xi}$. Therefore, the hypothesis of
Theorem 54 in \cite{1998Sontag} (p.p. 476) hold and the existence of a maximal
solution $\xi(t)$ of (\ref{eqn:ksi_dot}) on a time interval $\left[
0,\tau_{\max}\right)  $ such that $\xi(t)\in\Omega_{\xi}$, $\forall
t\in\left[  0,\tau_{\max}\right)  $ is ensured.

\textit{Phase B. }We have proven in \emph{Phase A} that $\xi(t)\in\Omega_{\xi}$,
$\forall t\in\left[  0,\tau_{\max}\right)  $ and more specifically that:%
\begin{equation}
\left.
\begin{array}
[c]{c}%
\xi_{d_{i}}(t)=\frac{e_{d_{i}}(t)}{\rho_{d_{i}}\left(  t\right)  }%
\in(-\underline{M}_{d_{i}},\overline{M}_{d_{i}})\\
\xi_{\beta_{i}}(t)=\frac{e_{\beta_{i}}(t)}{\rho_{\beta_{i}}\left(  t\right)
}\in(-\underline{M}_{\beta_{i}},\overline{M}_{\beta_{i}})
\end{array}
\right\}  \text{, }i=1,\dots,N \label{eqn:satisfied_for_tmax}%
\end{equation}
for all $t\in\left[  0,\tau_{\max}\right)  $, from which we obtain that
$e_{d_{i}}(t)$ and $e_{\beta_{i}}(t)$ are absolutely bounded by $\max
\{\underline{M}_{d_{i}},\overline{M}_{d_{i}}\}$ and $\max\{\underline
{M}_{\beta_{i}},\overline{M}_{\beta_{i}}\}$ respectively for { $i=1,\dots,N$.}
Let us also define:%
\begin{equation}
r_{d_i}\left(  \xi_{d_i}\right)  =  \tfrac{\frac
{1}{\underline{M}_{d_{i}}}+\frac{1}{\overline{M}_{d_{i}}}}{\left(  1+\frac
{\xi_{d_{i}}}{\underline{M}_{d_{i}}}\right)  \left(  1-\frac{\xi_{d_{i}}%
}{\overline{M}_{d_{i}}}\right)  }\text{, }{i=1,\dots,N}  \text{.}
\label{eqn:r_d}%
\end{equation}
Now, assume there exists a set $I\subseteq\{1,\ldots,N\}$ such that
$\lim_{t\rightarrow\tau_{\max}}\xi_{d_{k}}(t)$\ = $\overline{M}_{d_{k}}$ (or
$-\underline{M}_{d_{k}}$), $\forall k\in I$. Hence, invoking
(\ref{eqn:epsilon_d}) and (\ref{eqn:r_d}), we conclude that $\lim
_{t\rightarrow\tau_{\max}}\varepsilon_{d_{k}}\left(  \xi_{d_{k}}(t)\right)
=+\infty$ (or $-\infty$) and $\lim_{t\rightarrow\tau_{\max}}r_{d_{k}}\left(
\xi_{d_{k}}(t)\right)  =+\infty$, $\forall k\in I$. Moreover, we also deduce from
(\ref{eqn:v_des}) that $\lim_{t\rightarrow\tau_{\max}}v_{k}\left(  \xi_{d_{k}%
},t\right)  $ remains bounded for all $k\in\bar{I}$, where $\bar{I}$ is the
complementary set of $I$. To proceed, let us define $\bar{k}$ $=\min\{I\}$ and
notice that $\varepsilon_{d_{\bar{k}}}\left(  \xi_{d_{\bar{k}}}\right)  $, as
derived from (\ref{eqn:epsilon_d}), is well defined for all $t\in\left[
0,\tau_{\max}\right)  $, owing to (\ref{eqn:satisfied_for_tmax}). Therefore,
consider the positive definite and radially unbounded function $V_{d_{\bar{k}%
}}=\frac{1}{2}\varepsilon_{d_{\bar{k}}}^{2}$ for which it is clear that
$\lim_{t\rightarrow\tau_{\max}}V_{d_{\bar{k}}}\left(  t\right)  =+\infty$.
However, differentiating $V_{d_{\bar{k}}}$with respect to time and
substituting (\ref{eqn:error_dynamics}), we obtain:%
\begin{align}
\dot{V}_{d_{\bar{k}}}  &  =\varepsilon_{d_{\bar{k}}}r_{d_{\bar{k}}}\left(
\xi_{d_{\bar{k}}}\right)  \rho_{d_{\bar{k}}}^{-1}(-k_{d_{\bar{k}}%
}\varepsilon_{\bar{k}}\cos\beta_{\bar{k}}\nonumber\\
&  \left.  \text{ \ \ }+v_{\bar{k}-1}\cos(\gamma_{\bar{k}}+\beta_{\bar{k}%
})-\dot{\rho}_{d_{\bar{k}}}\xi_{d_{\bar{k}}}\right)  \text{.} \label{eqn:dVdk}%
\end{align}
from which, owing to the fact that $v_{\bar{k}-1}\cos(\gamma_{\bar{k}}%
+\beta_{\bar{k}})-\dot{\rho}_{d_{\bar{k}}}\xi_{d_{\bar{k}}}$ is bounded and
$\cos\left(  \beta_{\bar{k}}\right)  >\cos\left(  \beta_{\mathsf{con}}\right)>0
$, we conclude that $\lim_{t\rightarrow\tau_{\max}}\dot{V}_{d_{\bar{k}}%
}(t)=-\infty$, which clearly contradicts to our supposition that $\lim_{t\rightarrow
\tau_{\max}}V_{d_{\bar{k}}}(t)=+\infty$. Thus, we conclude that $\bar{k}$ doesn't
exist and hence that $I$ is an empty set. Therefore, there exist
$\underline{\xi}_{d_{i}}$ and $\bar{\xi}_{d_{i}}$ such that:
\begin{equation}
-\underline{M}_{d_{i}}<\underline{\xi}_{d_{i}}\leq\xi_{d_{i}}(t)\leq\bar{\xi
}_{d_{i}}<\overline{M}_{d_{i}}{,\ }\forall t\in\left[  0,\tau_{\max}\right)
\label{eqn:ksi_d_final}%
\end{equation}
for all {$i=1,\dots,N$,} from which it can be easily deduced that
$\varepsilon_{d}\left(  \xi_{d}\right)  $ and consequently the control input
(\ref{eqn:v_des}) remain bounded for all $t\in\left[  0,\tau_{\max}\right)  $.

Notice also from{ (\ref{eqn:satisfied_for_tmax}) that }$\varepsilon_{\beta
}\left(  \xi_{\beta}\right)  $, as derived from (\ref{eqn:epsilon_beta}), is
well defined for all $t\in\left[  0,\tau_{\max}\right)  $. Therefore, consider
the positive definite and radially unbounded function $V_{\beta}=\frac{1}%
{2}\varepsilon_{\beta}^{T}K_{\beta}^{-1}\varepsilon_{\beta}$. Differentiating
$V_{\beta}$ with respect to time and substituting (\ref{eqn:ksi_beta_dot_2}),
we obtain:%
\begin{align*}
\dot{V}_{\beta} &  =-\left\Vert \varepsilon_{\beta}^{T}r_{\beta}\left(
\xi_{\beta}\right)  (\rho_{\beta}\left(  t\right)  )^{-1}\right\Vert
^{2}+\varepsilon_{\beta}^{T}r_{\beta}\left(  \xi_{\beta}\right)  (\rho_{\beta
}\left(  t\right)  )^{-1}K_{\beta}^{-1}\\
&  \left.  \text{ \ \ \ \ \ \ \ }(D^{-1}\tilde{S}K_{d}\varepsilon_{d}\left(  \xi_{d}\right)
+D^{-1}s-\dot{\rho}_{\beta}\left(  t\right)  \xi_{\beta})\right.  \text{.}%
\end{align*}
Hence, exploiting the boundedness of $D^{-1}$, $\tilde{S}$, $s$ and
$\varepsilon_{d}\left(  \xi_{d}\right)  $, we get:%
\begin{align}
\dot{V}_{\beta}\leq & -\left\Vert \varepsilon_{\beta}^{T}r_{\beta}\left(
\xi_{\beta}\right)  (\rho_{\beta}\left(  t\right)  )^{-1}\right\Vert
^{2}\nonumber\\
&+\left\Vert \varepsilon_{\beta}^{T}r_{\beta}\left(  \xi_{\beta}\right)
(\rho_{\beta}\left(  t\right)  )^{-1}\right\Vert K_{\beta}^{-1}\bar{B}_{\beta
}\label{eqn:dVb}%
\end{align}
where $\bar{B}_{\beta}$ is a positive constant independent of $\tau_{\max}$,
satisfying:%
\begin{equation}
\left\Vert D^{-1}(\tilde{S}K_{d}\varepsilon_{d}\left(  \xi_{d}\right) +s -D\dot{\rho}_{\beta}\left(
t\right)  \xi_{\beta})\right\Vert {\leq\bar{B}}_{\beta}\label{eqn:B_beta}%
\end{equation}
for all $\xi(t)\in\Omega_{\xi}$. Therefore, we conclude that $\dot{V}_{\beta}$
is negative when $\left\Vert \varepsilon_{\beta}^{T}r_{\beta}\left(
\xi_{\beta}\right)  (\rho_{\beta}\left(  t\right)  )^{-1}\right\Vert
>K_{\beta}^{-1}\bar{B}_{\beta}$, from which, owing to the positive
definiteness and diagonality of $r_{\beta}\left(  \xi_{\beta}\right)
(\rho_{\beta}\left(  t\right)  )^{-1}$ and $K_{\beta}^{-1}$ as well as
employing (\ref{eqn:PP_rho}) and (\ref{eqn: r_beta}), it can be easily
verified that:%
\[
\left\Vert \varepsilon_{\beta}(t)\right\Vert \leq\bar{\varepsilon}_{\beta
}:=\max\left\{  \left\Vert \varepsilon_{\beta}(0)\right\Vert ,K_{\beta}%
^{-1}{\bar{B}_{\beta}\max}\left\{  {\tfrac{\underline{M}_{\beta_{i}}%
\overline{M}_{\beta_{i}}}{\underline{M}_{\beta_{i}}+\overline{M}_{\beta_{i}}}%
}\right\}  \right\}
\]
for all $t\in\left[  0,\tau_{\max}\right)  $. Furthermore, invoking the
inverse logarithm in (\ref{eqn:epsilon_beta}), we obtain:%
\begin{equation}%
\begin{array}
[c]{l}%
-\underline{M}_{\beta_{i}}<\tfrac{e^{-\bar{\varepsilon}_{\beta}}-1}%
{e^{-\bar{\varepsilon}_{\beta}}+1}\underline{M}_{\beta_{i}}=\text{
\ \ \ \ \ \ \ \ \ \ \ \ \ \ \ \ \ \ \ \ \ \ \ \ \ \ \ \ \ \ \ \ \ }\\
\multicolumn{1}{c}{\underline{\xi}_{\beta_{i}}\leq\xi_{\beta_{i}}(t)\leq
\bar{\xi}_{\beta_{i}}}\\
\multicolumn{1}{r}{=\tfrac{e^{\bar{\varepsilon}_{\beta}}-1}{e^{\bar
{\varepsilon}_{\beta}}+1}\overline{M}_{\beta_{i}}<\overline{M}_{\beta_{i}}}%
\end{array}
\label{eqn:ksi_beta_final}%
\end{equation}
{for all $t\in\left[  0,\tau_{\max}\right)  $ and $i=1,\dots,N$. Thus, the
control input $\omega\left(  \xi_{\beta},t\right)  $ designed in
(\ref{eqn:wmega_des}) remains bounded for all $t\in\left[  0,\tau_{\max
}\right)  $. }

Up to this point, what remains to be shown is that $\tau_{\max}$ can be
extended to $\infty$. In this direction, notice by (\ref{eqn:ksi_d_final}) and
(\ref{eqn:ksi_beta_final}) that $\xi(t)\in\Omega_{\xi}^{\prime}=\Omega
_{\xi_{d}}^{\prime}\times\Omega_{\xi_{\beta}}^{\prime}${, {$\forall
t\in\left[  0,\tau_{\max}\right)  $, where }%
\begin{align*}
\Omega_{\xi_{d}}^{\prime}  &  =[\underline{\xi}_{d_{1}},\bar{\xi}_{d_{1}%
}]\times\ldots\times\lbrack\underline{\xi}_{d_{N}},\bar{\xi}_{d_{N}}]\\
\Omega_{\xi_{\beta}}^{\prime}  &  =[\underline{\xi}_{\beta_{1}},\bar{\xi
}_{\beta_{1}}]\times\ldots\times\lbrack\underline{\xi}_{\beta_{N}},\bar{\xi
}_{\beta_{N}}]
\end{align*}
are nonempty and compact subsets of $\Omega_{\xi_{d}}$ and $\Omega_{\xi
_{\beta}}$\ respectively. Hence, assuming that $\tau_{\max}<\infty$ and since
$\Omega_{\xi}^{\prime}\subset$ $\Omega_{\xi}$, Proposition C.3.6 in \cite{1998Sontag} (p.p. 481)
dictates the existence of a time instant $t^{\prime}\in\left[  0,\tau_{\max
}\right)  $ such that $\xi(t^{\prime})\notin\Omega_{\xi}^{\prime}$, which is a
clear contradiction. Therefore, $\tau_{\max}=\infty$ and $\xi(t)\in\Omega
_{\xi}^{\prime}\subset$ $\Omega_{\xi}$, {$\forall t\geq0$. }Finally,
multiplying (\ref{eqn:ksi_d_final}) and (\ref{eqn:ksi_beta_final}) by
$\rho_{d_{i}}\left(  t\right)  $ and $\rho_{\beta_{i}}\left(  t\right)  $
respectively, we conclude:}%
\begin{equation}
\left.
\begin{array}
[c]{c}%
-\underline{M}_{d_{i}}\rho_{d_{i}}\left(  t\right)  <e_{d_{i}}\left(
t\right)  <\overline{M}_{d_{i}}\rho_{d_{i}}\left(  t\right) \\
-\underline{M}_{\beta_{i}}\rho_{\beta_{i}}\left(  t\right)  <e_{\beta_{i}%
}\left(  t\right)  <\overline{M}_{\beta_{i}}\rho_{\beta_{i}}\left(  t\right)
\end{array}
\right\}  \text{, }\forall t\geq0 \label{eqn:e_final}%
\end{equation}
{for all $i=1,\dots,N$ and consequently the solution of the 2-D formation control
problem with prescribed performance under collision and connectivity
constraints for the considered platoon of vehicles. }\begin{figure}[ptb]
\begin{center}
\includegraphics[width=3.35in]{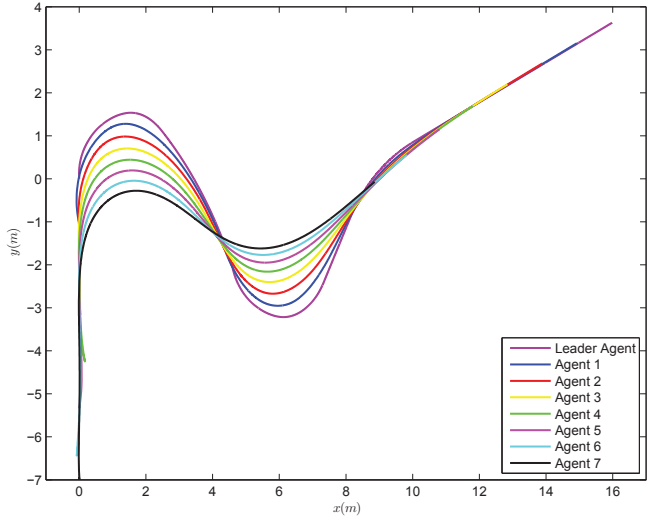}
\end{center}
\caption{The trajectories on a planar surface of the vehicles composing the
platoon.}%
\label{Fig:Trajectories}%
\end{figure}
\end{proof}

\begin{remark}
From the aforementioned proof it can be deduced that the proposed control
scheme achieves its goals without resorting to the need of rendering the
transformed errors $\varepsilon_{d}\left(  \xi_{d}\right)  $, $\varepsilon
_{\beta}\left(  \xi_{\beta}\right)  $ arbitrarily small by adopting extreme
values of the control gains $K_{d}$, $K_{\beta}$ (see (\ref{eqn:dVdk}) and
(\ref{eqn:dVb})). The actual performance given in (\ref{eqn:e_final}) is
solely determined by the designer-specified functions $\rho_{d_{i}}\left(
t\right)  ,\rho_{\beta_{i}}\left(  t\right)  $ and parameters $\underline
{M}_{d_{i}},\overline{M}_{d_{i}},\underline{M}_{\beta_{i}},\overline{M}%
_{\beta_{i}}$, that are related to the collision and connectivity constraints.
Furthermore, the selection of the control gains $K_{d}$, $K_{\beta}$ is
significantly simplified to adopting those values that lead to reasonable
control effort. Nonetheless, is should be noted that their selection affects
the control input characteristics (i.e., decreasing the gain values leads to
increased oscillatory behavior within the prescribed performance envelope described by (\ref{eqn:PP_inequalities}), which
is improved when adopting, higher values, enlarging, however, the control
effort both in magnitude and rate). Additionally, fine tuning might be needed
in real-time scenarios, to retain the required linear and angular velocities
within the range that can be implemented by the motors. Similarly, control input constraints impose an upper bound on the required
speed of convergence of $\rho_{d_{i}}\left(  t\right)$, $\rho_{\beta_{i}%
}\left(  t\right)$ that is affected by the exponentials $e^{-l_{d}t}$, $e^{-l_{\beta}t}$.
\end{remark}

\section{Simulation Results \label{sec:simulation_results}}

To demonstrate the efficiency of the proposed decentralized
control protocol, a realistic simulation was carried out in the
WEBOTS$^{\text{TM}}$ platform \cite{Webots}, considering a platoon
comprising of a Pioneer3AT/leader and $7$ Pioneer3DX following vehicles. The inter-vehicular distance and the bearing angle are obtained by a camera with range $D=2$m and angle of view $AoV=90^{o}$, that is mounted on each
Pioneer3DX vehicle and detects a white spherical marker attached on its predecessor. The leading vehicle performs a smooth maneuver depicted in
Fig. \ref{Fig:Trajectories}, along with the trajectories of the following
vehicles. The desired distance between successive vehicles is
set equally at $d_{i,des}=d=0.75$m$,${ \ $i=1,\dots,7$, whereas the collision
and connectivity constraints are given b}y $d_{\mathsf{col}}=0.05d=0.0375$m
and $d_{\mathsf{con}}=D=2$m. Regarding the heading error, we select
$\beta_{\mathsf{con}}=\frac{AoV}{2}=45^{o}$. In addition, we require steady
state error of no more than $0.0625$m and minimum speed of convergence as
obtained by the exponential $e^{-0.5t}$ for the distance error. Thus, invoking
(\ref{eqn:Md_Mbeta}), we select the parameters $\underline{M}_{d_{i}%
}=0.7125$m$,\overline{M}_{d_{i}}=1.25$m and the functions $\rho_{d_{i}%
}\left(  t\right)  =(1-\frac{0.0625}{1.25})e^{-0.5t}+\frac{0.0625}{1.25},${
$i=1,\dots,7$. In the same vein, we require maximum steady state error of
}$1.15^{o}$ and minimum speed of convergence as obtained by the exponential
$e^{-0.5t}$ for the heading error. Therefore, $\underline{M}_{\beta_{i}%
}=\overline{M}_{\beta_{i}}=\beta_{\mathsf{con}}=45^{o}$ and $\rho_{\beta_{i}%
}\left(  t\right)  =(1-\frac{1.15}{45})e^{-0.5t}+\frac{1.15}{45},${
$i=1,\dots,7$. Finally, we chose $K_{d}=\mathrm{diag}%
[0.005_{,}\dots,0.005]$ and $K_{\beta}=\mathrm{diag}[0.001_{,}\dots,0.001]$ to
produce reasonable linear and angular velocities that can be implemented by the motors of the mobile robots.}

The simulation results are illustrated in Figs. \ref{Fig:Distance_errors}%
-\ref{Fig:Distances}. More specifically, the evolution of the distance and
heading errors $e_{d_{i}}\left(  t\right)  $, $e_{\beta_{i}}\left(  t\right)$, $i=1,\dots,7$ is depicted in Figs. \ref{Fig:Distance_errors} and
\ref{Fig:Orientation_errors} respectively, along with the corresponding
performance bounds. The inter-vehicular distance along with the collision and
connectivity constraints are pictured in Fig. \ref{Fig:Distances}. As it was
predicted by the theoretical analysis, the decentralized 2-D control problem of vehicular platoons under limited visual feedback is solved
with guaranteed transient and steady state response, collision
avoidance and connectivity maintenance. Finally, the accompanying video demonstrates the aforementioned simulation study in the WEBOTS$^{\text{TM}}${
platform.

\begin{figure}[ptb]
\begin{center}
\includegraphics[width=3.35in]{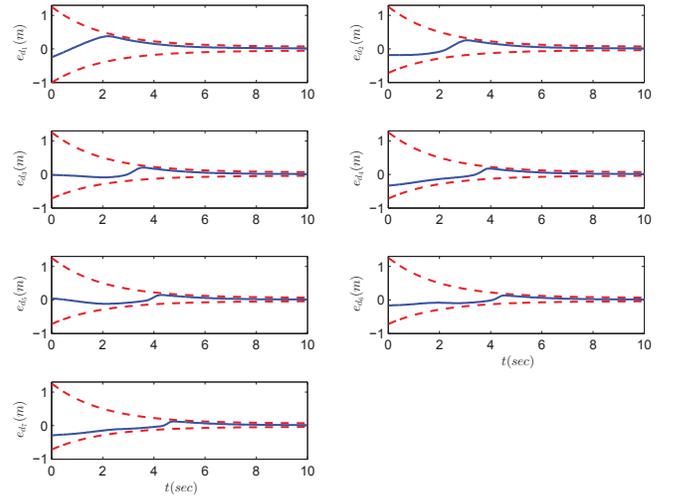}
\end{center}
\caption{The evolution of the distance errors $e_{d_{i}}(t),i=1,...,7$ (blue
lines), along with the imposed performance bounds (red lines).}%
\label{Fig:Distance_errors}%
\end{figure}

\begin{figure}[ptb]
\begin{center}
\includegraphics[width=3.35in]{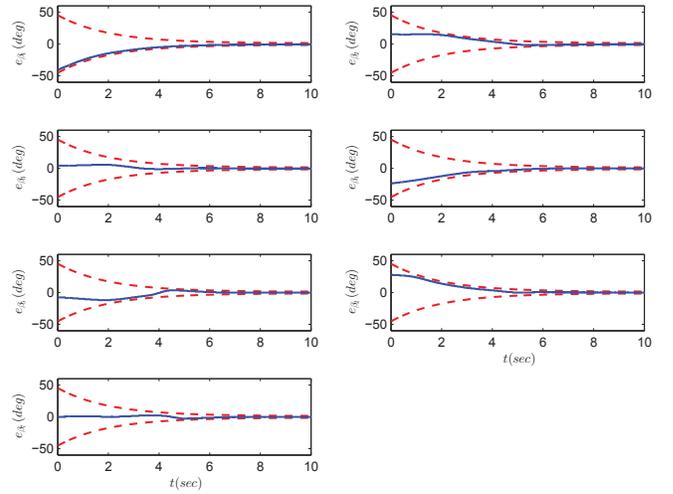}
\end{center}
\caption{The evolution of the heading errors $e_{\beta_{i}%
}(t),i=1,...,7$ (blue lines), along with the imposed performance bounds (red
lines).}%
\label{Fig:Orientation_errors}%
\end{figure}

\begin{figure}[ptb]
\begin{center}
\includegraphics[width=3.35in]{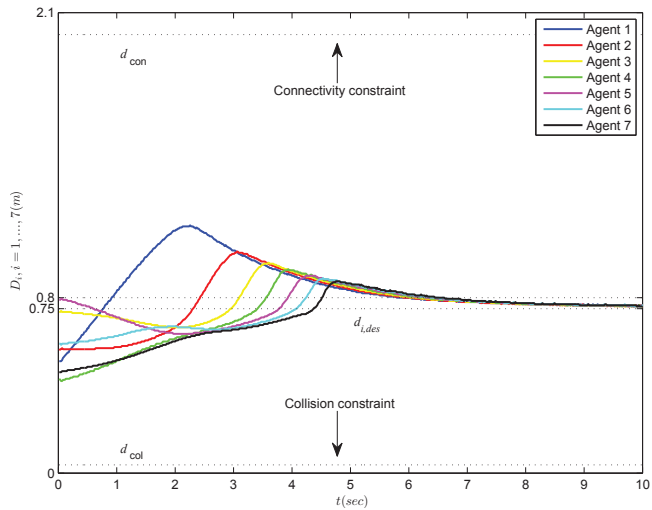}
\end{center}
\caption{The distance between successive vehicles along with the collision and
connectivity constraints.}%
\label{Fig:Distances}%
\end{figure}

\section{Conclusions \label{sec:conclusions}}

We proposed a $2$-D decentralized control protocol for vehicular
platoons under the predecessor-following architecture, that establishes
arbitrarily fast and maintains with arbitrary accuracy a desired formation
without: i) any inter-vehicular collisions and ii) violating the connectivity
constraints imposed by the limited field of view of the onboard cameras that
are used for visual feedback. Future research efforts will be devoted towards:
i) addressing the bidirectional architecture in a similar framework (i.e.,
prescribed performance as well as collision and connectivity
constraints), ii) guaranteeing obstacle avoidance and iii) extending the
control protocol to apply for uncertain nonlinear vehicle dynamics. Finally, real-time
experiments will be conducted to verify the theoretical findings.

\bibliographystyle{IEEEtran}
\bibliography{2D_Platoon}

% Generated by IEEEtran.bst, version: 1.13 (2008/09/30)
\begin{thebibliography}{10}
\providecommand{\url}[1]{#1}
\csname url@samestyle\endcsname
\providecommand{\newblock}{\relax}
\providecommand{\bibinfo}[2]{#2}
\providecommand{\BIBentrySTDinterwordspacing}{\spaceskip=0pt\relax}
\providecommand{\BIBentryALTinterwordstretchfactor}{4}
\providecommand{\BIBentryALTinterwordspacing}{\spaceskip=\fontdimen2\font plus
\BIBentryALTinterwordstretchfactor\fontdimen3\font minus
  \fontdimen4\font\relax}
\providecommand{\BIBforeignlanguage}[2]{{%
\expandafter\ifx\csname l@#1\endcsname\relax
\typeout{** WARNING: IEEEtran.bst: No hyphenation pattern has been}%
\typeout{** loaded for the language `#1'. Using the pattern for}%
\typeout{** the default language instead.}%
\else
\language=\csname l@#1\endcsname
\fi
#2}}
\providecommand{\BIBdecl}{\relax}
\BIBdecl

\bibitem{1996Swaroop_Hedrick}
D.~Swaroop and J.~Hedrick, ``String stability of interconnected systems,''
  \emph{Automatic Control, IEEE Transactions on}, vol.~41, no.~3, pp. 349--357,
  Mar 1996.

\bibitem{2001Soo_Chong_Roh}
T.~S. No, K.-T. Chong, and D.-H. Roh, ``A lyapunov function approach to
  longitudinal control of vehicles in a platoon,'' \emph{Vehicular Technology,
  IEEE Transactions on}, vol.~50, no.~1, pp. 116--124, Jan 2001.

\bibitem{2005Jovanovic_Bamieh}
M.~R. Jovanovic and B.~Bamieh, ``On the ill-posedness of certain vehicular
  platoon control problems,'' \emph{IEEE Transactions on Automatic Control},
  vol.~50, no.~9, pp. 1307--1321, 2005.

\bibitem{2009Barooah_Mehta_Hespanha}
P.~Barooah, P.~G. Mehta, and J.~P. Hespanha, ``Mistuning-based control design
  to improve closed-loop stability margin of vehicular platoons,'' \emph{IEEE
  Transactions on Automatic Control}, vol.~54, no.~9, pp. 2100--2113, 2009.

\bibitem{2012Lin_Fardad_Jovanovic}
F.~Lin, M.~Fardad, and M.~Jovanovic, ``Optimal control of vehicular formations
  with nearest neighbor interactions,'' \emph{Automatic Control, IEEE
  Transactions on}, vol.~57, no.~9, pp. 2203--2218, Sept 2012.

\bibitem{1994Hedrick_Tomizuka_Varaiya}
J.~Hedrick, M.~Tomizuka, and P.~Varaiya, ``Control issues in automated highway
  systems,'' \emph{Control Systems, IEEE}, vol.~14, no.~6, pp. 21--32, Dec
  1994.

\bibitem{1997Li_Roberto_et_al}
P.~Y. Li, R.~Horowitz, L.~Alvarez, J.~Frankel, and A.~M. Robertson, ``An
  automated highway system link layer controller for traffic flow
  stabilization,'' 1997.

\bibitem{2000_Rajamani}
R.~Rajamani, H.-S. Tan, B.~K. Law, and W.-B. Zhang, ``Demonstration of
  integrated longitudinal and lateral control for the operation of automated
  vehicles in platoons,'' \emph{Control Systems Technology, IEEE Transactions
  on}, vol.~8, no.~4, pp. 695--708, Jul 2000.

\bibitem{1998Tan_Rajesh_Zhang}
H.~. Tan, R.~Rajesh, and W.~. Zhang, ``Demonstration of an automated highway
  platoon system,'' in \emph{Proceedings of the American Control Conference},
  vol.~3, 1998, pp. 1823--1827.

\bibitem{2004_Mazo_Speranazon}
M.~Mazo, A.~Speranzon, K.~Johansson, and X.~Hu, ``Multi-robot tracking of a
  moving object using directional sensors,'' in \emph{Robotics and Automation,
  2004. Proceedings. ICRA '04. 2004 IEEE International Conference on}, vol.~2,
  April 2004, pp. 1103--1108 Vol.2.

\bibitem{2005_Mariottini_Pappas}
G.~Mariottini, G.~Pappas, D.~Prattichizzo, and K.~Daniilidis, ``Vision-based
  localization of leader-follower formations,'' in \emph{Decision and Control,
  2005 and 2005 European Control Conference. CDC-ECC '05. 44th IEEE Conference
  on}, Dec 2005, pp. 635--640.

\bibitem{2005_Gustavi}
T.~Gustavi and X.~Hu, ``Formation control for mobile robots with limited sensor
  information,'' in \emph{Robotics and Automation, 2005. ICRA 2005. Proceedings
  of the 2005 IEEE International Conference on}, April 2005, pp. 1791--1796.

\bibitem{2002_Das_Fierro}
A.~Das, R.~Fierro, V.~Kumar, J.~Ostrowski, J.~Spletzer, and C.~Taylor, ``A
  vision-based formation control framework,'' \emph{Robotics and Automation,
  IEEE Transactions on}, vol.~18, no.~5, pp. 813--825, Oct 2002.

\bibitem{2008_Gustavi_Hu}
T.~Gustavi and X.~Hu, ``Observer-based leader-following formation control using
  onboard sensor information,'' \emph{Robotics, IEEE Transactions on}, vol.~24,
  no.~6, pp. 1457--1462, Dec 2008.

\bibitem{2006_Khatir_Davison}
M.~Khatir and E.~Davison, ``A decentralized lateral-longitudinal controller for
  a platoon of vehicles operating on a plane,'' in \emph{American Control
  Conference, 2006}, June 2006, pp. 6 pp.--.

\bibitem{2004_Pham_Wang}
M.~Pham and D.~Wang, ``A unified nonlinear controller for a platoon of car-like
  vehicles,'' in \emph{American Control Conference, 2004. Proceedings of the
  2004}, vol.~3, June 2004, pp. 2350--2355 vol.3.

\bibitem{2013Ali_Garcia_Martinet}
A.~Ali, G.~Garcia, and P.~Martinet, ``Minimizing the inter-vehicle distances of
  the time headway policy for urban platoon control with decoupled longitudinal
  and lateral control,'' in \emph{Intelligent Transportation Systems - (ITSC),
  2013 16th International IEEE Conference on}, Oct 2013, pp. 1805--1810.

\bibitem{2004_Tanner_Pappas}
H.~Tanner, G.~Pappas, and V.~Kumar, ``Leader-to-formation stability,''
  \emph{Robotics and Automation, IEEE Transactions on}, vol.~20, no.~3, pp.
  443--455, June 2004.

\bibitem{2003_Lawton_Beard}
J.~Lawton, R.~Beard, and B.~Young, ``A decentralized approach to formation
  maneuvers,'' \emph{Robotics and Automation, IEEE Transactions on}, vol.~19,
  no.~6, pp. 933--941, Dec 2003.

\bibitem{2014_Maithripala_Berg}
D.~Maithripala, J.~Berg, D.~Maithripala, and S.~Jayasuriya, ``A geometric
  virtual structure approach to decentralized formation control,'' in
  \emph{American Control Conference (ACC), 2014}, June 2014, pp. 5736--5741.

\bibitem{1994Godbole_Lygeros}
D.~Godbole and J.~Lygeros, ``Longitudinal control of the lead car of a
  platoon,'' \emph{Vehicular Technology, IEEE Transactions on}, vol.~43, no.~4,
  pp. 1125--1135, Nov 1994.

\bibitem{2003_Vidal}
R.~Vidal, O.~Shakernia, and S.~Sastry, ``Formation control of nonholonomic
  mobile robots with omnidirectional visual servoing and motion segmentation,''
  in \emph{Robotics and Automation, 2003. Proceedings. ICRA '03. IEEE
  International Conference on}, vol.~1, Sept 2003, pp. 584--589 vol.1.

\bibitem{2009_Hao_Barooah}
H.~Hao, P.~Barooah, and P.~G. Mehta, ``Distributed control of two dimensional
  vehicular formations: stability margin improvement by mistuning,'' in
  \emph{ASME 2009 Dynamic Systems and Control Conference}.\hskip 1em plus 0.5em
  minus 0.4em\relax American Society of Mechanical Engineers, 2009, pp.
  699--706.

\bibitem{2014_Panagou_Kumar}
D.~Panagou and V.~Kumar, ``Cooperative visibility maintenance for
  leader-follower formations in obstacle environments,'' \emph{Robotics, IEEE
  Transactions on}, vol.~30, no.~4, pp. 831--844, Aug 2014.

\bibitem{2003_Cowan_et_al}
N.~Cowan, O.~Shakerina, R.~Vidal, and S.~Sastry, ``Vision-based
  follow-the-leader,'' in \emph{Intelligent Robots and Systems, 2003. (IROS
  2003). Proceedings. 2003 IEEE/RSJ International Conference on}, vol.~2, Oct
  2003, pp. 1796--1801 vol.2.

\bibitem{2007_Mariottini_et_al}
G.~Mariottini, F.~Morbidi, D.~Prattichizzo, G.~Pappas, and K.~Daniilidis,
  ``Leader-follower formations: Uncalibrated vision-based localization and
  control,'' in \emph{Robotics and Automation, 2007 IEEE International
  Conference on}, April 2007, pp. 2403--2408.

\bibitem{2009_Mariottini_et_al}
G.~Mariottini, F.~Morbidi, D.~Prattichizzo, N.~Vander~Valk, N.~Michael,
  G.~Pappas, and K.~Daniilidis, ``Vision-based localization for leader-follower
  formation control,'' \emph{Robotics, IEEE Transactions on}, vol.~25, no.~6,
  pp. 1431--1438, Dec 2009.

\bibitem{2014Bechlioulis_Dimarogonas_Kyriakopoulos}
C.~P. Bechlioulis, D.~V. Dimarogonas, and K.~J. Kyriakopoulos, ``Robust control
  of large vehicular platoons with prescribed transient and steady state
  performance,'' in \emph{Proceedings of the IEEE Conference on Decision and
  Control}, 2014, Accepted.

\bibitem{2014Bechlioulis_Rovithakis}
C.~P. Bechlioulis and G.~A. Rovithakis, ``A low-complexity global
  approximation-free control scheme with prescribed performance for unknown
  pure feedback systems,'' \emph{Automatica}, vol.~50, no.~4, pp. 1217--1226,
  2014.

\bibitem{1998Sontag}
E.~D. Sontag, \emph{Mathematical Control Theory}.\hskip 1em plus 0.5em minus
  0.4em\relax London, U.K.: Springer, 1998.

\bibitem{Webots}
{Cyberbotics}, ``{Webots: A commercial Mobile Robot Simulation Software},''
  http://www.cyberbotics.com, {Online}.

\end{thebibliography}

\end{document}